\documentclass[letterpaper,11pt]{article}
\usepackage[utf8]{inputenc}

\usepackage{amsmath, amsthm, amssymb, thm-restate}
\usepackage{algorithmicx}
\usepackage[table,xcdraw]{xcolor}
\usepackage[vlined,linesnumbered]{algorithm2e}
\usepackage{setspace}
\usepackage{mathtools}
\usepackage[numbers]{natbib}
\usepackage{comment} 
\usepackage{tcolorbox} 
\usepackage{xfrac}
\usepackage{hyperref}
\usepackage{multirow}
\usepackage{caption}
\usepackage{bm}
\usepackage{newfloat}
\usepackage{enumitem}
\usepackage{dblfloatfix} 
\usepackage{wrapfig}
\usepackage{csquotes}
\usepackage[noabbrev,nameinlink]{cleveref}

\usepackage{svg}
\usepackage{physics}
\usepackage{amsmath}
\usepackage{tikz}
\usepackage{mathdots}
\usepackage{yhmath}
\usepackage{cancel}
\usepackage{color}
\usepackage{siunitx}
\usepackage{array}
\usepackage{multirow}
\usepackage{amssymb}
\usepackage{gensymb}
\usepackage{tabularx}
\usepackage{extarrows}
\usepackage{booktabs}
\usetikzlibrary{fadings}
\usetikzlibrary{patterns}
\usetikzlibrary{shadows.blur}
\usetikzlibrary{shapes}

\usepackage{fancyhdr}

\usepackage[margin=1in]{geometry}

\allowdisplaybreaks

\definecolor{mygreen}{RGB}{10,110,230}
\definecolor{myred}{RGB}{10,110,230}

\hypersetup{
     colorlinks=true,
     citecolor= mygreen,
     linkcolor= myred
}

\renewcommand{\epsilon}{\varepsilon}

\newcommand{\hiddencomment}[1]{}

\renewcommand{\O}{O}
\newcommand{\tO}{\widetilde{\O}}

\DeclareMathOperator*{\Exp}{\ensuremath{{\mathbb{E}}}}
\DeclareMathOperator*{\Prob}{\ensuremath{\textnormal{Pr}}}
\renewcommand{\Pr}{\Prob}

\newcommand{\floor}[1]{{\left\lfloor{#1}\right\rfloor}}
\newcommand{\prob}[1]{\Pr\paren{#1}}
\newcommand{\card}[1]{\left\lvert#1\right\rvert}
\newcommand{\paren}[1]{\left( #1 \right)}
\newcommand{\bracket}[1]{\left[ #1 \right]}
\newcommand{\expect}[1]{\Exp\bracket{#1}}

\newcommand{\eps}{\varepsilon}

\renewcommand{\epsilon}{\eps}

\newcounter{protocolcounter}
\newenvironment{protocol}{
\refstepcounter{protocolcounter}
\begin{whitetbox}
\textbf{Protocol \theprotocolcounter:}%
}{\end{whitetbox}} 
\crefname{protocolcounter}{Protocol}{Protocols}

\newcounter{lpcounter}
 
\crefname{lpcounter}{LP}{LPs}

\DeclareMathOperator{\poly}{poly}

\crefname{lemma}{Lemma}{Lemmas}
\crefname{theorem}{Theorem}{Theorems}
\crefname{property}{Property}{Properties}
\crefname{claim}{Claim}{Claims}
\crefname{definition}{Definition}{Definitions}
\crefname{observation}{Observation}{Observations}
\crefname{proposition}{Proposition}{Propositions}
\crefname{assumption}{Assumption}{Assumptions}
\crefname{line}{Line}{Lines}
\crefname{figure}{Figure}{Figures}
\creflabelformat{property}{(#1)#2#3}
\crefname{equation}{}{}
\crefname{section}{Section}{Sections}
\crefname{appendix}{Appendix}{Appendices}
\crefname{algCounter}{Algorithm}{Algorithms}
\Crefname{algCounter}{Algorithm}{Algorithms}

\newtheorem{theorem}{Theorem}
\newtheorem{lemma}{Lemma}[section]
\newtheorem{proposition}[lemma]{Proposition}

\newtheorem{definition}[lemma]{Definition}
\newtheorem{claim}[lemma]{Claim}

\definecolor{mylightgray}{RGB}{230,230,230}

\algnewcommand{\IIf}[2]{\textbf{if} #1 \textbf{then} #2}
\algnewcommand{\EndIIf}{\unskip\ \algorithmicend\ \algorithmicif}

\newenvironment{graytbox}{
\par\addvspace{0.1cm}
\begin{tcolorbox}[width=\textwidth,
                  boxsep=5pt,
                  left=1pt,
                  right=1pt,
                  top=2pt,
                  bottom=2pt,
                  boxrule=0pt,
                  arc=0pt,
                  colback=mylightgray,
                  colframe=black,
                  ]
}{
\end{tcolorbox}
}

\newenvironment{whitetbox}{
\par\addvspace{0.1cm}
\begin{tcolorbox}[width=\textwidth,
                  boxsep=5pt,
                  left=1pt,
                  right=1pt,
                  top=2pt,
                  bottom=2pt,
                  boxrule=1pt,
                  arc=0pt,
                  colframe=black,
                  colback=white
                  ]
}{
\end{tcolorbox}
}

\newcounter{algCounter}

\makeatletter
\renewcommand{\paragraph}{%
  \@startsection{paragraph}{4}%
  {\z@}{10pt}{-1em}%
  {\normalfont\normalsize\bfseries}%
}
\makeatother

\makeatletter
\patchcmd{\@algocf@start}
  {-1.5em}
  {0pt}
  {}{}
\makeatother

\title{Robust Communication Complexity of Matching:\\ EDCS Achieves 5/6 Approximation}

\author{
Amir Azarmehr\\{\em Northeastern University} \and Soheil Behnezhad \\{\em Northeastern University}
}

\date{}

\begin{document}

\maketitle

\thispagestyle{empty}

\begin{abstract}
    We study the {\em robust communication complexity} of maximum matching. Edges of an {\em arbitrary} $n$-vertex graph $G$ are {\em randomly} partitioned between Alice and Bob independently and uniformly. Alice has to send a single message to Bob such that Bob can find an (approximate) maximum matching of the whole graph $G$. We specifically study the best approximation ratio achievable via protocols where Alice communicates only $\tO(n)$ bits to Bob.

    \smallskip\smallskip    
    There has been a growing interest on the robust communication model due to its connections to the random-order streaming model. An algorithm of Assadi and Behnezhad [ICALP'21] implies a $(2/3+\epsilon_0 \sim .667)$-approximation for a small constant $0 < \epsilon_0 < 10^{-18}$, which remains the best-known approximation for general graphs. For bipartite graphs, Assadi and Behnezhad [Random'21] improved the approximation to .716 albeit with a computationally inefficient (i.e., exponential time) protocol.
    
    \smallskip\smallskip    
    In this paper, we study a natural and efficient protocol implied by a random-order streaming algorithm of Bernstein [ICALP'20] which is based on {\em edge-degree constrained subgraphs} (EDCS) [Bernstein and Stein; ICALP'15]. The result of Bernstein immediately implies that this protocol achieves an (almost) $(2/3 \sim .666)$-approximation in the robust communication model. We present a new analysis, proving that it achieves a much better (almost) $(5/6 \sim .833)$-approximation. This significantly improves previous approximations both for general and bipartite graphs. We also prove that our analysis of Bernstein's protocol is tight.
    
\end{abstract}


\clearpage
\setcounter{page}{1}

\section{Introduction}

Given an $n$-vertex graph $G=(V, E)$, a {\em matching} is a collection of vertex disjoint edges in $G$ and a {\em maximum matching} is the matching with the maximum size. In this paper, we study matchings in Yao's (one-way) communication model \cite{Yao79}. The edge-set $E$ is partitioned between two players Alice and Bob. Alice has to send a single message to Bob such that Bob can find an (approximate) maximum matching of the whole graph $G$. We are particularly interested in the trade-off between the size of the message sent by Alice and the approximation ratio of the output solution. Besides being a natural problem, this communication model is closely related to streaming algorithms and has thus been studied extensively over the years \cite{GoelKK12,Kapralov21,ChakrabartiCM08,AssadiB-ICALP21,AssadiB-RANDOM21}.

In order to obtain an {\em exact} maximum matching, it is known that $\Omega(n^2)$ bits of communication are needed \cite{FeigenbaumKMSZ05}. That is, the trivial protocol where Alice sends her whole input to Bob is optimal. The situation is more interesting for approximate solutions. It is clear that $\Omega(n)$ words of communication are needed for any approximation as the whole matching can be given to Alice. A natural question, therefore, studied in numerous prior works \cite{GoelKK12,Kapralov13,Kapralov21,AssadiB-ICALP21,AssadiB-RANDOM21}, is the best approximation achievable via protocols that have a near-optimal communication complexity of $\widetilde{O}(n) = \O(n \poly\log)$.

It is not hard to see that if Alice sends a maximum matching of her input to Bob, then Bob can find a 1/2-approximate matching. There is, however, a more sophisticated approach based on the powerful {\em edge-degree constrained subgraph (EDCS)} of \citet*{BernsteinS15} that achieves an (almost) 2/3-approximation (see the paper of \citet*{AssadiBern-SOSA19}). This turns out to be the right approximation under an adversarial partitioning of edges. In their seminal paper, \citet*{GoelKK12} proved that obtaining a better than 2/3-approximation requires $n^{1+1/(\log \log n)} \gg n \poly\log n$ communication.

The communication model discussed above is {\em doubly worst-case} in that both the input graph and the edge partitioning are chosen by an adversary. In this paper, we study the so called {\em robust communication} model---\mbox{\`a la} \citet*{ChakrabartiCM08}---where the graph $G$ is still chosen by an adversary but its edges are now {\em randomly} partitioned between Alice and Bob (i.e., each edge is uniformly given either to Alice or Bob independently). This model goes beyond the doubly worst-case scenario discussed above and sheds light on whether the hardness of a problem is inherent to the input graph or rather a pathological partitioning of its edges. Another motivation behind the study of the robust communication model is its connections to {\em random-order} streams. In particular, almost all known lower bounds for random-order streams are proved in this robust communication model.

While existing adversarial partitioning protocols already imply an (almost) 2/3-approximation in the robust communication model, a random-order streaming algorithm of \citet*{AssadiB-ICALP21} implies a better bound. Their algorithm starts with an EDCS-based algorithm of \citet*{Bernstein20}, and then augments it with a number of short {\em augmenting paths}, achieving a $(2/3+\epsilon_0)$-approximation for some fixed constant $0 < \epsilon_0 < 10^{-18}$.
This remains the best-known approximation in general graphs. For bipartite graphs, an entirely different approach of \citet*{AssadiB-RANDOM21} achieves a larger .716-approximtaion although their protocol runs in doubly exponential time.

In this paper, we give a new analysis for the EDCS-based protocol of \citet*{Bernstein20} showing that, without any augmentation, it  already achieves a much better than 2/3-approximation. 
\begin{graytbox}
\begin{theorem}\label{thm:main}
Bernstein's protocol \cite{Bernstein20} with high probability achieves a $(1-\epsilon)5/6 \sim .833$ approximation in the robust communication model using $O(n \cdot \log n \cdot \poly(1/\epsilon))$ words of communication.
\end{theorem}
\end{graytbox}

\cref{thm:main} improves, rather significantly, the state-of-the-art approximation for both general and bipartite graphs from $.667$ \cite{AssadiB-ICALP21}  and $.716$  \cite{AssadiB-RANDOM21} respectively to $.833$. We note that Bernstein's protocol runs in linear time in the input size; hence \cref{thm:main}, in addition to improving approximation, also improves the running time of the algorithm of \cite{AssadiB-RANDOM21} from doubly exponential to linear. Besides these quantitative improvements, we believe that a more important qualitative implication of \cref{thm:main} is that EDCS, which has been used in the literature to only obtain $2/3$ or slightly-larger-than-2/3 approximations in various models, can be used to obtain a significantly better approximation in the robust communication model.

Our analysis can be applied to the more general {\em multi-party} one-way robust communication model where instead of two players Alice and Bob, the input is randomly partitioned between $k$ players (see \cref{sec:preliminaries} for the formal definition of the model). This communication model is particularly of interest since any lower bound in it, for any choice of $k$, also implies a lower bound for random-order streams. We show the following, which generalizes \cref{thm:main}:


\begin{graytbox}
    \begin{theorem} \label{thm:multi-party}
    For any $k \geq 2$ and any $\epsilon > 0$, Bernstein's protocol \cite{Bernstein20} in the $k$-party one-way robust communication model achieves a $(1-\epsilon)(\frac{2}{3} + \frac{1}{3k})$-approximation of maximum matching using messages of length $O(n \cdot \log n \cdot \poly(1/\epsilon))$.
    \end{theorem}
\end{graytbox}

We note that the current best approximation known in the random-order streaming setting for maximum matching is $(2/3 + 10^{-18})$ by \citet*{AssadiB-ICALP21}. \cref{thm:multi-party} implies that either there is a better random-order streaming algorithm for maximum matching (which likely is the case), or else to prove a tight lower bound via the multi-party communication model, one has to consider at least $k \geq 10^{18}/3$ parties!

Finally, we show that our guarantees of \cref{thm:main,thm:multi-party} are tight for Bernstein's protocol. That is, we show that:

\begin{graytbox}
    \begin{theorem} \label{thm:tightness}
For any $k \geq 2$, there exist an infinite family of graphs $G$ such that the expected approximation ratio of Bernstein's protocol in the $k$-party one-way robust communication model
is at most $(\frac{2}{3} + \frac{1}{3k})$.
\end{theorem}
\end{graytbox}

\section{Technical Overview} \label{sec:technical-overview}

Bernstein's protocol constructs two subgraphs $H$ and $U$ of size $O(n \log n)$ both of which will be communicated  to Bob. Subgraph $H$ is constructed solely by Alice who does so by revealing only $\epsilon$ fraction of her input graph. The construction guarantees that for some sufficiently large constant $\beta \geq 1$, every edge $(u, v) \in H$ satisfies $\deg_H(u)+\deg_H(v) \leq \beta$. That is, $H$ has edge-degree upper bounded by $\beta$. This already implies that $H$ has at most $O(n\beta) =O(n)$ edges. The subgraph $U$ is simply the set of all the remaining edges $(u, v)$ in the graph $G$ (given either to Alice or Bob) for which $\deg_H(u) + \deg_H(v) \leq \beta -1$. In other words, all the remaining ``underfull'' edges whose edge-degree is less than $\beta$ are added to $U$. While it is not at all clear that $H$ can be constructed in such a way that guarantees $|U| = O(n \log n)$, \citet*{Bernstein20} showed this is indeed possible. At the end, Bob returns a maximum matching of all the edges that he receives.

The subgraph $H \cup U$ can be shown to include an {\em edge-degree constrained subgraph} (EDCS) of $G$, which is known to include a $(2/3-O(\epsilon))$-approximate maximum matching of the base graph $G$ for $\beta \geq 1 /\epsilon$ \cite{BernsteinS15,Behnezhad-EDCS-Arxiv}. This already implies an (almost) 2/3-approximation in our model. This guarantee is in fact tight for the maximum matching contained in $H \cup U$ as illustrated in \cref{fig:tight}. In the example of \cref{fig:tight}, the missed (red dashed) edges have edge-degree $\beta$ in $H$, and so they do not belong to $U$. While the graph in the example of \cref{fig:tight} has a perfect matching, any matching in $H \cup U$ can only match 2/3-fraction of vertices.

\begin{figure}
    \centering
    \includegraphics[scale=0.9]{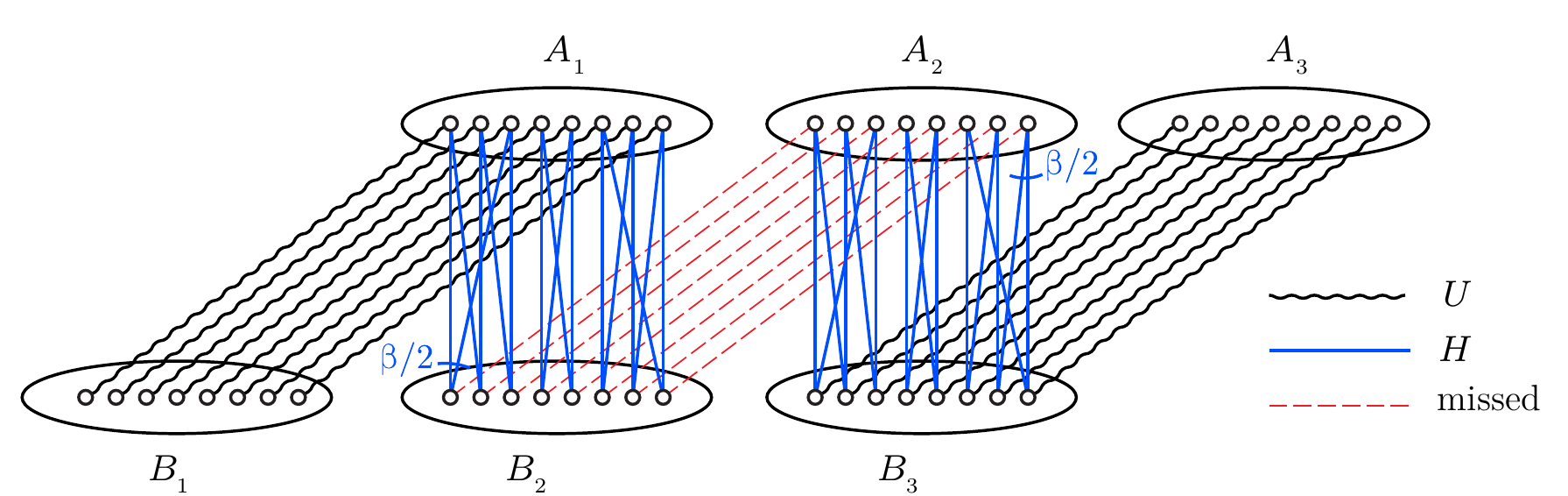}
    \caption{\small An example where  subgraph $H \cup U$ in Bernstein's protocol does not include a better than 2/3-approximation.}
    \label{fig:tight}
\end{figure}

The crucial insight is that although $H \cup U$ may only include a 2/3-approximate matching of the graph, Bob in addition will also have access to the set $E_B$ of the edges originally given to him in the random partitioning. So instead of $H \cup U$, we need to focus on the size of the maximum matching contained in $H \cup U \cup E_B$. Let us now revisit the example of \cref{fig:tight}. As we discussed, the set $H$ is only constructed using a small $\epsilon$ fraction of the edges. 
Moreover, conditioned on $H$, the subgraph $U$ will also be fully determined regardless of how the edges are partitioned between Alice and Bob. This implies that, {\em even conditioned on the outcome of $H$ and $U$}, each dashed edge is given to Bob with probability (almost) 1/2. This results in an (almost) 5/6-approximation in the example of \cref{fig:tight}: We can combine the 2/3-approximate black matching in $U$ with half of the dashed edges, obtaining an (almost) $\frac{2}{3} + \frac{1}{2} \cdot \frac{1}{3} = \frac{5}{6}$ approximation. We remark that this example already shows that our $5/6$-approximation guarantee of \cref{thm:main} is tight for Bernstein's protocol (see \cref{thm:tightness} for the formal proof).

The nice property of the example of \cref{fig:tight} is that subgraph $H \cup U$ includes a 2/3-approximate matching $M$ (the black matching in $U$) where removing its vertices from the graph still leaves a 1/3-approximate matching in $G$ (the dashed red edges). If we prove that this holds for every graph, then we immediately get an (almost) 5/6-approximation analysis for Bernstein's protocol. Unfortunately, however, this property does not hold for all graphs. In \cref{apx:bad-example}, we provide examples of $H, U$ such that for {\em every} matching $M$ in $H \cup U$, it holds that
$$
    |M| + \frac{1}{2}\mu(G - V(M)) \leq 0.75 \mu(G),
$$
where $\mu(G - V(M))$ here is the size of maximum matching remained in graph $G$ after removing vertices of $M$. This implies that this idea is not sufficient to guarantee an (almost) 5/6-approximation for Bernstein's protocol.

In our analysis, instead of first committing to a $2/3$-approximate matching in $H \cup U$ and then augmenting it using the edges in $E_B$, we first commit to a smaller $1/2$-approximate matching by fixing an arbitrary maximum matching $M^*$ and taking half of its edges that are given to Bob. The advantage of this smaller 1/2-approximate matching is that it can be augmented much better. Specifically, we show that this 1/2-approximate matching, in expectation, can be augmented by a matching of size (almost) $\mu(G)/3$ using the edges in $H \cup U$, achieving overall a matching of size (almost) $\frac{1}{2} \mu(G) + \frac{1}{3} \mu(G) = \frac{5}{6} \mu(G)$. The proof of why a matching of size $\mu(G)/3$ can be found within the available vertices is the crux of our analysis and is formalized via fractional matchings.

\section{Preliminaries}\label{sec:preliminaries}


We start by formally defining the robust communication model for maximum matching.

\begin{definition}
    In the $k$-party one-way robust communication model,
    each edge is assigned independently and uniformly to one of the parties.
    The $i$-th party, supplied with the assigned edges 
    and a message $m_i$ from the $(i-1)$-th party,
    decides what message to send to the $(i+1)$-th party.
    The $k$-th and last party is responsible for reporting a matching.
    The communication complexity of a protocol in this model,
    is defined as the maximum number of words in the messages communicated between the parties,
    i.e.\ $\max_i \card{m_i}$, where $\card{m_i}$ denotes the number of words in $m_i$.

    In case $k = 2$, we refer to the first party as Alice and to the second party as Bob.
\end{definition}



We use $\mu(G)$ to denote the size of the maximum matching in graph $G$. For an edge $e$, we define its edge-degree as the sum of the degrees of its endpoints.

\subsection*{Background on Matching Theory} 

\begin{proposition}[folklore] \label{prp:general-fractional-matching}
    Let $G$ be any graph,
    and let $x$ be a fractional matching on $G$,
    such that for every vertex set $S \subseteq V$
    that $\card{S}$ is smaller than $\frac{1}{\epsilon}$,
    we have
    $$
        \sum_{e \in G[S]} x_e \leq \floor{\frac{\card{S}}{2}}.
    $$
    Then, it holds that $\mu(G) \geq (1 - \epsilon) \sum_e x_e$.
\end{proposition}

\begin{proof}[Proof sketch]
Let $z$ be another fractional matching where $z_e = (1 -\epsilon)x_e$.
If the $x$ satisfies the blossom inequality, i.e.\ $\sum_{e \in G[S]} x_e \leq \floor{\frac{\card{S}}{2}}$,
for all $S$ of size at most $\frac{1}{\epsilon}$,
Then $z$ satisfies it for all $S$.
To see this,
let $S$ be an odd-sized vertex set size at least $\frac{1}{\epsilon}$
such that $\sum_{e \in G[S]} x_e \geq \floor{\frac{\card{S}}{2}}$.
Then it holds:
$$
\sum_{e \in G[S]} z_e = (1 - \epsilon)\sum_{e \in G[S]} x_e \leq \sum_{e \in G[S]} x_e - \frac{1}{2} \leq \frac{\card{S}}{2} - \frac{1}{2} \leq \floor{\frac{\card{S}}{2}}.
$$
Hence there exists an integral matching of size at least 
$\sum_e z_e = (1 - \epsilon) \sum_e x_e$.
\end{proof}

The following definitions were introduced by \citet*{Bernstein20}.
The proposition, from the same paper, plays a key role in our analysis.

\begin{definition}
    A graph $H$ has bounded edge-degree $\beta$,
    if for all edges $(u,v) \in E_H$ it holds that
    $d_H(u) + d_H(v) \leq \beta$.
\end{definition}

\begin{definition}
    Given a graph $G$, and a subgraph $H \subseteq G$ 
    an edge $(u, v) \in E_G \setminus E_H$
    is $(H, \beta, \lambda)$-underfull
    if $d_H(u) + d_H(v) < (1 - \lambda)\beta$.
\end{definition}

\begin{proposition} [Lemma 3.1 from \cite{Bernstein20}]
\label{prp:edcs-appx}
    Fix any $\epsilon \in \bracket{0, \frac{1}{2}}$, let $\lambda, \beta$ be parameters such that $\lambda \leq \frac{\epsilon}{384}$,
    $\beta \geq 50 \lambda^{-2} \log\paren{\frac{1}{\lambda}}$.
    Consider any graph $G$, and any subgraph $H$ with bounded edge-degree $\beta$.
    Let $U$ contain all the $(H, \beta, 3\lambda)$-underfull edges in $G \setminus H$.
    Then $\mu(H \cup U) \geq \paren{\frac{2}{3} - \epsilon}\mu(G)$.
\end{proposition}

\subsection*{Concentration Inequalities} 
We use the following concentration inequalities in our proofs.

\begin{proposition}[Chernoff bound] \label{chernoff}
Let $X_1, \ldots, X_n$ be independent random variables
taking values in $[0, 1]$.
Let $X = \sum X_i$ and let $\mu = \Exp\bracket{X}$.
Then, for any $0 < \delta \leq 1$ and $0 < a \leq \mu$, we have 
$$
\prob{X \geq (1 + \delta)\mu} \leq \exp\paren{-\frac{\delta^2 \mu}{3}} \qquad \text{and} \qquad \prob{X \geq \mu + a} \leq \exp\paren{-\frac{a^2}{3 \mu}}.
$$
\end{proposition}

\begin{definition}[\cite{BoucheronLM09}]
A function $f : \{0,1\}^n \to \mathbb{N}$ is self-bounding
if there exist functions $f_1, \ldots, f_n : \{0,1\}^{n-1} \to \mathbb{N}$
such that for all $x \in \{0,1\}^n$ satisfy
$$
0 \leq f(x) - f_i(x^{(i)}) \leq 1 \qquad \forall i \in [n], 
$$
and
$$
\sum_{i = 1}^n \paren{f(x) - f_i(x^{(i)})} \leq f(x).
$$
Where $x^{(i)}$ is obtained by dropping the $i$-th component of $x$.
\end{definition}

\begin{proposition}[\cite{BoucheronLM09}] \label{prp:self-bounding-concentration}
    Take a self-bounding function $f : \{0,1\}^n \to \mathbb{N}$,
    and independent $0-1$ variables $X_1, \ldots, X_n$.
    Define $Z = f(X_1, \ldots, X_n)$.
    Then, it holds that
    $$
        \prob{Z \leq \Exp Z - t} \leq \exp\paren{\frac{-t^2}{2 \Exp Z}}.
    $$
\end{proposition}

\newcommand{\Min}{M_{\textnormal{in}}}
\newcommand{\Mout}{M_{\textnormal{out}}}
\newcommand{\tG}{{\tilde{G}}}
\newcommand{\tH}{{\tilde{H}}}
\newcommand{\tU}{{\tilde{U}}}
\newcommand{\tM}{{\tilde{M}}}
\newcommand{\hy}{\hat{y}}

\section{A New Analysis of Bernstein's Protocol}
This section is devoted to the proof of \cref{thm:main,thm:multi-party}.
We will provide an analysis of Bernstein's protocol (\cref{main-protocol})
in the two-party model.
To make this analysis applicable to the multi-party model,
we assume that each edge is assigned to Bob independently with probability $p \leq \frac{1}{2}$. 

We give a description of the protocol for the two-party model here.
The multi-party protocol is rather similar and we describe it in
the \hyperref[prf:main2]{Proof of Theorem 2}.
Let $E_A$ be the set of edges assigned to Alice,
and $E_B$ be the set of edges assigned to Bob.
Also, fix a constant $\epsilon \in \bracket{0, \frac{1}{2}}$ and let $\lambda = \frac{\epsilon}{384}$, and $\beta = 50 \lambda^{-4}$. The protocol is formalized as \cref{main-protocol}.

\begin{figure}[h]
\begin{protocol} 
    Bernstein's protocol via EDCS in the two-party one-way robust communication model.
    \label{main-protocol}
    
    \vspace{0.4cm} 
    
    \textbf{Alice:}
    \begin{enumerate}
        \item Take a subsample $E_s$ that includes each edge of $E_A$ independently with probability $\frac{\epsilon}{1 - p}$.
        \item Take a subgraph $H$ of bounded edge-degree $\beta$ from $E_s$,
        such that the number of ${(H, \beta, \lambda)}$-underfull edges
        in $E_r = E(G) \setminus E_s$
        is $O(n \cdot \log n \cdot \poly(1/\epsilon))$
        with high probability. (See \cref{cl:few-underfull} for the existence of $H$.)
        \item Find the $(H, \beta, \lambda)$-underfull edges of $E_A \setminus E_s$, call them $U_A$.
        \item Communicate $H \cup U_A$ to Bob.
    \end{enumerate}
    
    \textbf{Bob:}
    \begin{enumerate}
        \item Return the maximum matching in $E_B \cup H \cup U_A$.
    \end{enumerate}
\end{protocol}
\end{figure}

\cref{cl:few-underfull}
shows that Alice can execute step 2,
and that \cref{main-protocol} has communication complexity $O(n \cdot \log n \cdot \poly(1/\epsilon))$.
Note that taking into account the randomization in dividing the edges
between Alice and Bob, $E_s$ can be considered a uniform sample from the whole edge set that contains each edge with probability $\epsilon$.

\begin{claim}[Lemma 4.1 in \cite{Bernstein20}]
\label{cl:few-underfull}
    Alice, by looking only at the edges of $E_s$, can take a subgraph $H \subseteq E_s$
    that has bounded degree $\beta$,
    and with high probability $E_r = E(G) \setminus E_s$ has at most $O(n \cdot \log n \cdot \poly(1/\epsilon))$
    many $(H, \beta, \lambda)$-underfull edges. 
\end{claim}

For the analysis,
we first construct a fractional matching $x$
of expected size $\paren{\frac{2}{3} - \O(\epsilon)}\mu(G)$, 
the support of which is contained in $H \cup U$.
Then we show that a fractional matching $y$
can be obtained from $x$,
such that its support is contained in $E_B \cup H \cup U_A$
and has expected size at least $\paren{\frac{2}{3} + \frac{p}{3} - \O(\epsilon)}\mu(G)$.
Finally, we use the structure of $y$
to show that its existence implies $E_B \cup H \cup U_A$
has an integral matching almost as large as the size of $y$.
In \cref{sec:whp}, we show that the approximation ratio
is also achieved with high probability.

First, we describe how to obtain the fractional matching $x$ given $H \cup U$,
where recall that $U$ is the set of $(H, \beta, \lambda)$-underfull edges in $E_r$.
Fix a maximum matching $M^*$ in $E_r$.
Let $\Min$ be the edges of $M^*$ that appear in $H \cup U$, i.e.\ $\Min = M^* \cap (H \cup U)$, and let $\Mout = M^* \setminus \Min$.
\begin{whitetbox}
\begin{itemize}[leftmargin=15pt]
    \item Start with $H_1 = H$, and $U_1 = U$.
    \item For $i = 1,\, \ldots,\, \lambda \beta:$
    \item \qquad Let $M_i$ be a maximum matching in $H_i \cup U_i$.
    \item \qquad Let $H_{i+1} = H_i \setminus (M_i \setminus \Min)$,
    $U_{i+1} = U_i \setminus (M_i \setminus \Min)$.
    \item For every edge $e$, let 
    $\displaystyle x_e = \frac{|\{i : e \in M_i \}|}{\lambda\beta}$.
\end{itemize}
\end{whitetbox}
One can think of this process as,
starting with $H \cup U$,
taking a maximum matching $M_i$ each time,
and removing $M_i \setminus \Min$
from the graph.
Then, letting $x_e$ equal to the fraction of matchings we have taken that include $e$.
Note that the matchings $M_1,\, \ldots\, M_{\lambda \beta}$
can intersect only in $\Min$.
We will use \cref{prp:edcs-appx} to show that $x$ has expected size at least $\paren{\frac{2}{3} - \O(\epsilon)}\mu(G)$.

\begin{lemma}\label{le:x-is-large}
    It holds that $\expect{\sum_e x_e} \geq \paren{\frac{2}{3} - \frac{5}{3}\epsilon}\mu(G)$.
\end{lemma}
\begin{proof}
    We apply \cref{prp:edcs-appx} to $G_i = (H \cup E_r) \setminus \paren{\bigcup_{j < i} M_j \setminus \Min}$, $H_i$, and $U_i$.
    After removing a matching from $H$, for any edge, its degree in $H$ will decrease by at most $2$.
    Also, $U$ contains all the $(H, \beta, \lambda)$-underfull edges of $E_r$.
    Hence, $U_i$ contains all the edges of $G_i \setminus H_i$
    that have $H_i$-degree smaller than $(1 - \lambda)\beta - 2(i - 1) \geq (1 - 3\lambda)\beta$.
    Therefore, \cref{prp:edcs-appx} implies
    $$
    \abs{M_i} \geq \paren{\frac{2}{3} - \epsilon}\mu(G_i).
    $$
    Also, notice that $G_i$ always includes $M^*$, consequently it holds $\mu(G_i) = \mu(E_r)$, and we have:
    $$
    \sum_e x_e \geq \frac{1}{\lambda \beta} \sum_i \card{M_i} \geq \paren{\frac{2}{3} - \epsilon}\mu(E_r) 
    $$
    Taking into account the fact that $\expect{\mu(E_R)} \geq (1 - \epsilon) \mu(G)$,
    we get:
    $$
     \expect{\sum_e x_e} 
     \geq \paren{\frac{2}{3} - \epsilon}(1 - \epsilon)\mu(G)
     \geq \paren{\frac{2}{3} - \frac{5}{3}\epsilon}\mu(G).\qedhere
    $$
\end{proof}


To describe how $y$ is obtained, we condition on $E_s$, thereby fixing $H$, $U$, and $x$.
The support of $y$ is included in $E_B \cup H \cup U_A$,
i.e.\ the edges that Bob will have access to in the end.
We show that $y$ has expected size at least about $p\cdot \mu(E_r) + (1-p)\sum_e x_e$,
where the randomness is over how the remaining edges $E_r$ are divided between Alice and Bob.
Lifting the condition on $E_s$,
the expectation of this value, by \cref{le:x-is-large}, is larger than $\paren{\frac{2}{3} + \frac{p}{3} - \O(\epsilon)}\mu(G)$.

After drawing $E_B$, take a matching $M'$,
which includes each edge of $\Min$ independently with probability $p$,
and includes each edge of $\Mout \cap E_B$ independently with probability $1 - \epsilon$.
Note, that conditioned on $E_s$, each edge of $\Mout$ is assigned to Bob with probability $\frac{p}{1 - \epsilon}$.
Hence, each edge of $\Mout$ ends up in $M'$
with probability $\frac{p}{1 - \epsilon}(1 - \epsilon) = p$,
i.e.\ $M'$ includes each edge of $M^*$ independently with probability $p$.

For any edge $e \notin M^*$, define $p_e$
as the probability of $e$
not being adjacent to any edge in $M'$.
Notice, $p_e$ is simply equal to
$(1 - p)$ to the power of the number of edges in $M'$
that are adjacent to $e$.
We define matching $\hy$ as follows:
\begin{equation*}
\hy_e =
\begin{cases}
1 &  \qquad\textnormal{if $e \in M'$,}  \\
x_e & \qquad\textnormal{if $e \in M^* \setminus M'$,}  \\
0 & \qquad\textnormal{if $e \notin M^*$ and $e$ is adjacent to an edge of $M'$,}  \\
\displaystyle (1 - p) \cdot \frac{x_e}{p_e} & \qquad\textnormal{otherwise.}
\end{cases}
\end{equation*}
We then scale down $\hy$ by a factor of $1 + \epsilon$,
and zero out some edges to obtain a fractional matching.
Formally, we let:
\begin{equation*}
    y_{(u,v)} = \begin{cases}
        0 & \textnormal{if $\hy_u / (1 + \epsilon) > 1$ or $\hy_v / (1 + \epsilon) > 1$}, \\
        \displaystyle \frac{\hy_{(u,v)}}{1 + \epsilon} & \textnormal{otherwise.}
    \end{cases}
\end{equation*}

\begin{lemma}
    Conditioned on $E_s$,
    it holds that $$\expect{\sum_e y_e}
        \geq (1 - 3\epsilon) p \cdot \mu(E_r) 
        + (1 - 3\epsilon)(1 - p) \sum_e x_e
        - 2\epsilon\mu(G).$$
\end{lemma}

\begin{proof}
    All the arguments made in this proof are conditioned on $E_s$.

    \begin{claim} \label{cl:expect-hy}
        For every vertex $u$, it holds that $\expect{\hy_u} = p \cdot \chi_{M^*}(u) + (1 - p)x_u$, where $\chi_{M^*}(u)$ is equal to $1$ if $u$ is covered by $M^*$ and zero otherwise.
    \end{claim}
    \begin{proof}
        First, consider a vertex $u$ that is covered by $M^*$,
        say by edge $e^* \in M^*$.
        When $e^*$ appears in $M'$,
        we have $\hy_{e^*} = 1$, and for all the other edges $e$ adjacent to $u$,
        the value of $\hy_e$ is equal to zero. Thus, we will have $\hy_u = 1$,
        i.e.\ $\expect{\hy_u \mid e^* \in M^*} = 1$.

        Now, we condition on $e^* \notin M'$.
        In this case, we will have $\hy_{e^*} = x_{e^*}$.
        Also, for any other edge $e$ adjacent to $u$,
        the probability that $e$ is not adjacent to any edge in $M'$
        is equal to $\frac{p_e}{1 - p}$. 
        Thus with probability $\frac{p_e}{1 - p}$ it holds that $\hy_e = (1 - p)\cdot\frac{x_e}{p_e}$,
        and we will have $\hy_e = 0$ otherwise. Hence, we can write:
        $$
        \expect{\hy_u \mid e^* \notin M'} = x_{e^*} + \sum_{\substack{e \ni u \\ e \neq e^*}} \frac{p_e}{1 - p} \cdot \paren{(1-p)\cdot\frac{x_e}{p_e}} = \sum_{e \ni u} x_e = x_u.
        $$
        Therefore, for a vertex $u$ that is covered by $M^*$, it holds that
        $$
        \expect{\hy_u} = p \cdot \expect{\hy_u \mid e^* \in M'} + (1 - p) \cdot \expect{\hy_u \mid e^* \notin M'}
        = p + (1 - p)x_u.
        $$

        The case where the vertex $u$ is not covered by $M^*$ follows similarly.
        For each edge $e$ adjacent to $u$ we have $\hy_e = (1 - p) \cdot \frac{x_e}{p_e}$
        with probability $p_e$, and we have $\hy_e = 0$ otherwise. Thus
        $$
        \expect{\hy_u} = \sum_{e \ni u} p_e \cdot \paren{(1 - p)\cdot \frac{x_e}{p_e}} = (1 - p)x_u. \qedhere
        $$
    \end{proof}

    The following claim helps us show that we do not lose much of $\hy$
    when we scale it down and zero out some of the edges.

    \begin{claim} \label{cl:low-overflow}
        For every vertex $u$, we have $\hy_u \leq 1$ if $u$ is not covered by $M^*$ or $x_u \leq \frac{1}{2}$, otherwise it holds that $\prob{\hy_u > 1 + \epsilon} \leq \epsilon$.
    \end{claim}
    \begin{proof}
        Consider a vertex $u$ not covered by $M^*$.
        For each edge $e$ adjacent to $u$ it holds that $p_e \geq 1 - p$ because $e$ has at most one neighbouring edge in $M^*$.
        Hence we have:
        $$
        \hy_u \leq \sum_{e \ni u} (1 - p) \cdot \frac{x_e}{p_e} 
        \leq \sum_{e \ni u} x_e
        \leq 1.
        $$

        Now take a vertex $u$ that is covered by $M^*$, say by edge $e^* \in M^*$.
        If $e^* \in M'$, then we have $\hy_u = 1$.
        Therefore we assume $e^* \notin M'$, and accordingly $\hy_e = x_e$.
        For any other edge $e$ it holds that $p_e \geq (1 - p)^2$.
        Hence we have:
        $$
        \hy_u \leq x_{e^*} + \sum_{\substack{e \ni u \\ e \neq e^*}} (1 - p) \cdot \frac{x_e}{p_e} 
        \leq x_{e^*} + \sum_{\substack{e \ni u \\ e \neq e^*}} \frac{x_e}{1 - p} 
        \leq 2 \sum_{e \ni u} x_e = 2 x_u.
        $$
        Thus, if it holds that $x_u \leq \frac{1}{2}$, then it follows $\hy_u \leq 1$.

        For the other cases,
        we use the \hyperref[chernoff]{Chernoff bound}
        to show that with high probability $\hy_u$ is not much larger than $1$.
        As mentioned before, if $e^*$ appears in $M'$, it holds that $\hy_u = 1$.
        Therefore, we condition on $e^* \notin M'$.
        We express $X = y_u - y_{e^*}$ as a sum of independent random variables that take values in $[0, 4/\lambda\beta]$.
        Note that since the edges outside $\Min$
        appear in at most one $M_i$,
        for $e \notin M^*$ we have $x_e \leq \frac{1}{\lambda\beta}$.

        Take an edge $e = (u, v) \neq e^*$.
        If $v$ is not matched in $M^*$ to another neighbour of $u$,
        then the value of $y_e$ is independent of the value of the other edges adjacent to $u$. It is equal to $(1 - p) \cdot \frac{x_e}{p_e} \leq \frac{2}{\lambda \beta}$ with probability $\frac{p_e}{1 - p}$, and zero otherwise.

        If $v$ is matched in $M^*$ to another neighbour $v'$ of $u$.
        Let $e' = (u, v')$. 
        Then the value of $y_e + y_e'$
        is independent of the value of the other edges adjacent to $u$.
        It is equal to $(1 - p) \cdot \frac{x_e + x_e'}{(1 - p)^2} \leq \frac{4}{\lambda \beta}$,
        with probability $(1 - p)$, and zero otherwise.

        Thus, by pairing the edges that are matched together,
        we can express $X$ as a sum of independent random variables in $[0, 4/\lambda \beta]$.
        The expectation of $X$, as calculated in \cref{cl:expect-hy}, is equal to $x_u - x_{e^*} \leq 1$. From the \hyperref[chernoff]{Chernoff bound} we get:
        \begin{align*}
        \prob{\hy_u > 1 + \epsilon}
            &\leq \prob{X > \Exp X + \epsilon} \\
            & \leq \prob{X \cdot \frac{\lambda \beta}{4} > 
            \mu\frac{\lambda \beta}{4} + \epsilon\frac{\lambda \beta}{4}} \\
            & \leq \exp\paren{-\frac{\epsilon^2 \lambda^2 \beta^2 / 16}{3 \mu \lambda \beta  /4}} \tag{By \hyperref[chernoff]{Chernoff bound}, noting that $X\lambda \beta / 4$ is a sum of independent random variables in $[0, 1].$}\\
            & = \exp\paren{-\frac{\epsilon^2 \lambda \beta}{12}} \\
            & < \epsilon\tag{Since $\lambda = \frac{\epsilon}{384}$ and $\beta = 50 \lambda^{-4}$},
        \end{align*}
        concluding the proof.
    \end{proof}

    We analyze $\expect{y_u}$.
    Consider generating $y_u$, in two steps.
    First, for every edge $(u, v)$,
    we let $y_{(u, v)}$ be equal to $\frac{1}{1 + \epsilon}\hy_{(u, v)}$ if $\hy_v \leq 1 + \epsilon$, and zero otherwise.
    Then, we zero out $y$ for all the edges adjacent to $u$ if $\hy_u > 1 + \epsilon$.

    In the first step, we lose a factor $(1 + \epsilon)$ when we scale $\hy_u$ down.
    Also, by \cref{cl:low-overflow}, 
    when we zero out edge $(u, v)$ because $y_v > 1 + \epsilon$, we lose
    an $\epsilon$-fraction from each edge, and consequently from $\expect{\hy_u}$.
    In the second step, again by \cref{cl:low-overflow},
    if $x_u \leq \frac{1}{2}$ or $u$ is not covered by $M^*$ we lose nothing.
    Otherwise, we zero out all the edges with probability at most $\epsilon$.
    We have $\hy_u \leq 2x_u \leq 2$, hence we lose an additive factor of $2\epsilon$.
    Overall for any vertex $u$ we get:
    $$
    \expect{y_u} \geq (1 - \epsilon) \frac{\expect{\hy_u}}{1 + \epsilon} - 2\epsilon
    \geq (1 - 3\epsilon) \paren{p\cdot\chi_{M^*}(u) + (1-p)x_u} - 2\epsilon,
    $$
    and for any vertex $u$ with $x_u \leq \frac{1}{2}$, we get:
    $$
    \expect{y_u} \geq (1 - 3\epsilon) \paren{p\cdot\chi_{M^*}(u) + (1-p)x_u}.
    $$
    Notice that since the sum of the components of $x$ is at most $\mu(G)$,
    there are at most $2 \mu(G)$ vertices with $x_u \geq \frac{1}{2}$.
    Thus, by summing the last two equations over $u$ we get:
    \begin{align*}
        \sum_u \expect{y_u} &=
        \sum_{u \in V(M^*)} \expect{y_u}
        + \sum_{u \notin V(M^*)} \expect{y_u} \\
        &\geq \paren{(1 - 3\epsilon)\sum_{u \in V(M^*)} p + (1-p)x_u}
        + \paren{(1-3\epsilon)\sum_{u \notin V(M^*)} (1-p)x_u}
        - 2\epsilon \cdot 2\mu(G) \\
        &= (1 - 3\epsilon) 2p \cdot \card{M^*} + (1 - 3\epsilon)(1-p)\sum_u x_u - 4\epsilon\mu(G).
    \end{align*}
    Recall that $\card{M^*} = \mu(E_r)$.
    Finally, by dividing both sides by $2$, we get:
    \begin{align*}
        \sum_e \expect{y_e} 
        &\geq (1 - 3\epsilon) p \cdot \mu(E_r) 
        + (1 - 3\epsilon)(1 - p) \sum_e x_e
        - 2\epsilon\mu(G). \qedhere
    \end{align*}
\end{proof}

Now we lift the condition on $E_s$.

\begin{lemma} \label{le:y-is-large}
    It holds that $\expect{\sum_e y_e} \geq \paren{\frac{2}{3} + \frac{p}{3} - 6\epsilon}\mu(G)$.
\end{lemma}
\begin{proof}
    We have:
    \begin{align*}
        \expect{\sum_e y_e}
        &= \expect{\expect{\sum_e y_e\ \bigg\vert \ E_s}} \\
        &\geq \expect{(1 - 3\epsilon) p \cdot \mu(E_r) 
        + (1 - 3\epsilon)(1 - p) \sum_e x_e
        - 2\epsilon\mu(G)} \\
        &\geq (1 - 3\epsilon) p \cdot (1 - \epsilon)\mu(G)
        + (1 - 3\epsilon)(1 - p) \paren{\frac{2}{3} - \frac{5}{3}\epsilon}\mu(G)
        - 2\epsilon\mu(G) \\
        \tag{by $\expect{\mu(E_r)} = (1-\epsilon)\mu(G)$ and \cref{le:x-is-large}} \\
        &\geq \paren{\frac{2}{3} + \frac{p}{3} - \paren{\frac{7}{6}p + \frac{29}{6}} \epsilon} \mu(G) \\
        &\geq \paren{\frac{2}{3} + \frac{p}{3} - 6 \epsilon} \mu(G). \qedhere
    \end{align*}
\end{proof}

We show that $E_B \cup H \cup U_A$
has an integral matching almost as large as the size of $y$.

\begin{lemma} \label{le:large-integer-matching}
    There exists a matching of size $(1 - 3\epsilon) \sum_e y_e$ in $E_B \cup H \cup U_A$.
\end{lemma}
\begin{proof}
    Notice that for every edge $e$, except the edges of $M^*$ which is a matching,
    it holds that $y_e \leq \frac{4}{\lambda \beta} \leq \epsilon^3$.
    Therefore, for any vertex set $S \subseteq V$
    that $\card{S}$ is smaller than $\frac{1}{\epsilon}$,
    we have:
    $$
        \sum_{e \in G[S]} x_e
        = \sum_{e \in G[S] \cap M^*} x_e
            + \sum_{e \in G[S] \setminus M^*} x_e
        \leq  \card{G[S] \cap M^*} + \frac{1}{\epsilon^2} \epsilon^3
        \leq \floor{\frac{\card{S}}{2}} + \epsilon.
    $$
    Hence, we can apply \cref{prp:general-fractional-matching}
    to $(1 - 2\epsilon)y$, to get $\mu(E_B \cup H \cup U_A) \geq (1 - 3\epsilon) \sum_e y_e$.
\end{proof}

\begin{proof}[Proof of \cref{thm:main}] \label{prf:main}
    By \cref{cl:few-underfull},
    Bernstein's protocol (\cref{main-protocol})
    is implementable using only $O(n \cdot \log n \cdot \poly(1/\epsilon))$
    words of communication.

    By \cref{le:y-is-large} there exists
    a fractional matching $y$ of expected size $\paren{\frac{2}{3} + \frac{p}{3} - 6\epsilon}\mu(G)$.
    Putting this together with \cref{le:large-integer-matching},
    we can conclude \cref{main-protocol} achieves a $\paren{\frac{2}{3} + \frac{p}{3} - 9\epsilon}$ approximation ratio.
    To see this approximation ratio is also achieved with high probability,
    refer to \cref{sec:whp}.
    Finally, letting $p = \frac{1}{2}$ and rescaling $\epsilon$ proves the theorem.
\end{proof}

\begin{proof}[Proof of \cref{thm:multi-party}] \label{prf:main2}
    We need to adjust \cref{main-protocol} for the $k$-party model.
    The first party will sample each of its edges independently with probability $\epsilon/(1 - 1/k)$ to obtain $E_s$.
    It will then construct the subgraph $H \subseteq E_s$ with bounded edge-degree,
    and send it to the next party along with the $(H, \beta, \lambda)$-underfull edges.
    Each of the next parties, except the last,
    communicates the $(H, \beta, \lambda)$-underfull edges it has been assigned
    along with the edges in the message it has received,
    to the next party.
    Finally, the last party will report the maximum matching
    in the graph consisting of all the edges to which it has access.
    
    This way, setting $p = \frac{1}{k}$,
    the first $k - 1$ parties will act as Alice in our analysis,
    and the last party acts as Bob.
    Hence, by a similar argument as in the \hyperref[prf:main]{Proof of Theorem 1},
    \cref{main-protocol} achieves a 
    $(\frac{2}{3} + \frac{1}{3k} - 9\epsilon)$ approximation ratio,
    and a rescaling of $\epsilon$ proves the theorem.
\end{proof}
\section{From Expectation to High Probability} \label{sec:whp}
In this section, we show that with a slight modification,
Bernstein's protocol (\cref{main-protocol}) achieves the $\frac{5}{6}$-approximation
with high probability.
To do so, Alice should send all the edges to Bob
when the number of edges is too small.


\begin{claim} 
\label{cl:large-matching}
    Without loss of generality, we can assume $\mu(G) = \Omega\paren{\log n}$.
\end{claim}
\begin{proof}
    A charging argument can be used to show that
    the number of edges in $G$ is less than $2 n \mu(G)$.
    Fix a maximum matching $M$ in $G$.
    For any edge $e$, 
    charge a unit to an edge of $M$ that is adjacent to $e$.
    Such an edge must exist since $M$ is a maximum matching.
    This way, we charge once for every edge in $G$,
    and every edge of $M$ is charged at most $n$ times
    through each of its endpoints.
    
    To see why the claim is true,
    note that in case $\mu(G)$ is too small,
    i.e.\ $\mu(G) = \O(\log n)$,
    the number of edges in the graph will be $\O(n \log n)$ and
    Alice can send all of its edges to Bob.
\end{proof}


\begin{lemma}
    Assuming that $\mu(G) = \Omega(\log n)$, whatever approximation ratio \cref{main-protocol} achieves in expectation,
    it will achieve with high probability.
\end{lemma}
\begin{proof}
We condition on the sample edge set $E_s$, thereby fixing $H$ and $U$.
Bob will have access to the edges of $H \cup U$
because they are either assigned to Bob, or they are communicated to Bob by Alice.
Let $e_1, \ldots, e_k$ be the other edges, i.e.\ the edges of $E_r \setminus U$.
Each of these edges is assigned to Bob independently with probability $\frac{1/2}{1 - \epsilon}$.

We define a self-bounding function $f : \{0, 1\}^k \to \mathbb{Z}$.
For $x \in \{0, 1\}^k$,
the value of $f(x)$ is equal to the maximum matching of $H \cup U \cup E_x$,
where $E_x = \{e_i \mid x_i = 1\}$.
Equivalently, $f(x)$ is the size of the output matching
when the edges $\{e_i \mid x_i = 1\}$ are assigned to Bob,
i.e.\ $E_B \cup H \cup U_A$ is equal to $E_x \cup H \cup U$.
Also, let $f_i(x^{(i)}) = f(x_1, \ldots x_{i-1}, 0, x_{i+1}, \ldots x_k)$.

Take any $x \in \{0, 1\}^k$.
Notice that $f_i(x^{(i)})$ is equal to $\mu(E_x \setminus e_i)$,
and removing an edge from a graph, will decrease its maximum matching by at most $1$.
Therefore, it holds:
$$
0 \leq f(x) - f_i(x^{(i)}) \leq 1, \qquad \forall i : 1 \leq i \leq k.
$$
Take the maximum matching $M$ in $H \cup U \cup E_x$,
and let $I$ be the indices of the edges in $E_x \cap M$,
i.e.\ $I = \{i \mid x_i = 1 \textnormal{ and } e_i \in M\}$.
For any $i \notin I$, the edge set $E_x \setminus e_i$ includes $M$.
Therefore, $f_i(x^{(i)})$ is equal to $f(x)$, and we have:
$$
\sum_{i = 1}^k \paren{f(x) - f_i(x^{(i)})} \leq \card{I} = f(x).
$$
Thus, $f$ is a self-bounding function.

We can now apply \cref{prp:self-bounding-concentration}.
Let $X_i$ be the indicator variable that $e_i$ is assigned to Bob,
i.e.\ $X_i$ is equal to $1$ when $e_i \in E_B$.
Let $Z = f(X_1, \ldots, X_k)$,
and $\mu = \expect{Z} = r\mu(G)$.
That is, $Z$ is the size of the output matching, and $r$ is the approximation ratio that \cref{main-protocol} achieves in expectation.
By \cref{prp:self-bounding-concentration},
we have:
$$
\prob{Z \leq r\mu(G) - \sqrt{2\mu(G) \log n}}
\leq \exp\paren{-\frac{2\mu(G) \log n}{2 \mu(G)}} = \frac{1}{n}.
$$
Thus, with high probability \cref{main-protocol}
outputs a matching of size $(1 - o(1)) r\mu(G)$.
Note that the deviation is $o(1)$ by the assumption that $\mu(G) = \Omega(\log n)$.
\end{proof}

\section{Some Instances for Bernstein's Protocol}\label{apx:bad-example}

In this section, we first prove \cref{thm:tightness} that our analysis of Bernstein's protocol in \cref{thm:main,thm:multi-party} are tight.
Then, we formalize a remark we made in \cref{sec:technical-overview}.

\begin{proof}[Proof of \cref{thm:tightness}]
    Consider a bipartite graph $G(L, R)$, such that $\card{L} = \card{R}$.
    Where $L$ consists of three equally-sized groups of vertices $A_1$, $A_2$, and $A_3$,
    and similarly $R$ consists of $B_1$, $B_2$, and $B_3$.
    The induced subgraphs $M_1 = G[A_1, B_1]$, $M_2 = G[A_2, B_2]$, and $M_3 = G[A_3, B_3]$
    are perfect matchings.
    The induced subgraphs $K_1 = G[A_1, B_2]$ and $K_2 = G[A_2, B_3]$ are complete bipartite graphs,
    and there are no other edges in the graph (see \cref{fig:tight}).
    Note that $G$ has a perfect matching, i.e.\ the size of the maximum matching is equal to $\card{L} = \card{R}$.
    Let $\beta \leq \frac{\card{V(G)}}{12 k}$ which is $O(\card{V(G)})$.

    Let $E_L$ be the set of edges assigned to the last party.
    To upper bound the approximation ratio of the multi-party protocol,
    we construct a vertex cover for $E_L \cup H \cup U$.
    It is a well-known fact that the size of the minimum vertex cover
    is larger than the size of the maximum matching.
    With high probability, the first party can take $H$ to be completely inside $K_1 \cup K_2$, so that $U$ will be equal the $M_1 \cup M_3$, and no edges of $M_2$ will appear in $H \cup U$.

    Let $X = V(M_2 \cap E_L) \cap A_2$, 
    i.e.\ $X$ includes one endpoint from every edge of $M_2$ that is assigned to the last party.
    We claim $A_1 \cup B_3 \cup X$ is a vertex cover for $E_L \cup H \cup U$.
    This is true because $A_1$ covers the edges of $M_1$ and $K_1$,
    $B_3$ covers the edges of $M_3$ and $K_2$,
    and $X$ covers all the remaining edges, which is $E_L \cap M_2$.
    
    Conditioned on the $H$ as described above, each edge of $M_2$ will be assigned to Bob with
    probability $\frac{1/k}{1 - \epsilon}$.
    Thus the expected size of the vertex cover is equal to
    $$
        \card{A_1} + \card{B_3} + \frac{1/k}{1 - \epsilon} \card{A_2}
        = \paren{\frac{1}{3} + \frac{1}{3} + \frac{1/k}{1-\epsilon} \cdot \frac{1}{3}}\mu(G)
        \leq (1 + 2\epsilon) \paren{\frac{2}{3} + \frac{1}{3k}} \mu(G).
    $$
    Letting $\epsilon$ be arbitrarily small proves the theorem.
\end{proof}

As mentioned in \cref{sec:technical-overview},
the graph discussed in \cref{thm:tightness} (see \cref{fig:tight})
has a nice property,
i.e.\ there exists a large matching $M$
such that $G - V(M)$ also has a large matching.
As the output of Bernstein's protocol 
has expected size of at least $\card{M} + \frac{1}{2}\mu(G - V(M))$,
which in this case is equal to $\frac{5}{6}\mu(G)$,
this property might seem useful to analyze the protocol.
However, the following claim shows that such an $M$ does not generally exist.

\begin{claim}
There exists a graph $G$, with arbitrarily large number of vertices,
such that for $\beta \leq \card{V(G)}/4$,
there is a choice of $H$ and $U$,
such that every matching $M$ in $H \cup U$
satisfies $|M| + \frac{1}{2}\mu(G - V(M)) \leq 0.75 \mu(G)$.
Where here $H$ is subgraph with bounded edge-degree $\beta$,
and $U$ is the set of the underfull edges in $G \setminus H$,
i.e.\ the edges of $G \setminus H$ with $H$-degree smaller than $\beta - 1$.
\end{claim}

\begin{proof}
Let $G(L, R)$ be a bipratite graph, such that $\card{L} = \card{R}$.
Where $L$ consists of four equally-sized groups of vertices $A_1$, $A_2$, $A_3$, and $A_4$,
and similarly $R$ consists of $B_1$, $B_2$, $B_3$, and $B_4$.
The induced subgraphs $M_1 = G[A_1, B_1]$,
$M_2 = G[A_2, B_2]$,
$M_3 = G[A_3, B_3]$,
and $M_4 = G[A_4, B_4]$
are perfect matchings.
The induced subgraphs $K_1 = G[A_1, B_2]$,
$K_2 = G[A_2, B_3]$,
and $K_3 = G[A_3, B_4]$
are complete bipartite graphs,
and there are no other edges in the graph.
Note that $G$ has a perfect matching (see \cref{fig:tight2}).

\begin{figure}[h]
    \centering
    \includegraphics[scale=0.75]{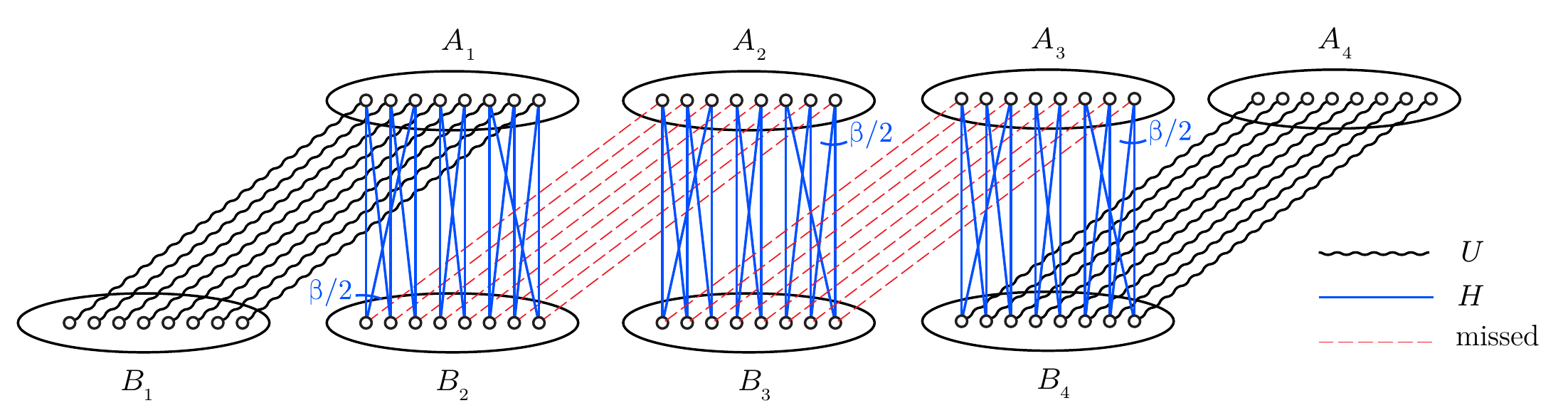}
    \caption{\small An example where every matching $M$ satisfies $|M| + \frac{1}{2}\mu(G - V(M)) \leq 0.75 \mu(G)$.}
    \label{fig:tight2}
\end{figure}

Let $H$ be a $(\beta/2)$-regular subgraph of $K_1 \cup K_2 \cup K_3$.
The corresponding $U$ is equal to $M_1 \cup M_4$,
and none of the edges in $M_2 \cup M_3$ appear in $H \cup U$.
We prove that $\max_M |M| + \frac{1}{2}\mu(G - V(M)) \leq 0.75 \mu(G)$,
where $M$ ranges over all the matchings in $H \cup U$.
We say that a matching is optimal if it achieves the maximum possible value for $|M| + \frac{1}{2}\mu(G - V(M))$.

First, we prove there exists an optimal matching that includes all of $M_1 \cup M_4$.
To see this, take an optimal matching $M$.
Take any vertex $u$ in $A_1$, and let its adjacent edge in $M_1$ be $e$.
If $u$ is covered by some edge $e' \in M$,
then removing $e'$ from $M$ and adding $e$ does not decrease 
$|M| + \frac{1}{2}\mu(G - V(M))$.
Because this would not change $\card{M}$
and can only increase $\mu(G - V(M))$.
Also, when $u$ is not covered by $M$,
adding the $e$ to $M$
will cause $\card{M}$ to grow by one,
and $\mu(G - V(M))$ to decrease by at most one.
A similar argument holds for the vertices in $B_4$.
Hence, by repeatedly adding such edges,
we can obtain an optimal matching containing $M_1 \cup M_4$.

Now we can restrict our attention to $A_2 \cup A_3 \cup B_2 \cup B_3$. 
We claim no matter what the rest of $M$ (a.k.a.\ $M \setminus (M_1 \cup M_2) = M \cap K_2$) is, the value of $\card{M} + \frac{1}{2}\mu(G - V(M))$ would be the same.
Because if $\card{M \cap K_2}$ is equal to $k$,
it holds that $\mu(G - V(M)) = \frac{1}{2}\mu(G) - 2k$.
Hence the optimal value of $|M| + \frac{1}{2}\mu(G - V(M))$
is equal to
$$
\card{M_1} + \card{M_2} + k + \frac{1}{2} \paren{\frac{1}{2}\mu(G) - 2k} = \frac{3}{4}\mu(G).
$$

To see why $\mu(G - V(M)) = \frac{1}{2}\mu(G) - 2k$,
note that since $M_1 \cup M_4 \subseteq M$,
any vertex of $B_2 \cup A_3$ is a singleton in $G - V(M)$.
Hence, a maximum matching in $G - V(M)$
is the set of edges in $M_2 \cup M_3$
that are not adjacent to an edge of $M$, which has size $\frac{1}{2}\mu(G) - 2k$.
\end{proof}


\section*{Acknowledgements}

The second author thanks David Wajc for enlightening discussions about going beyond 2/3-approximations via EDCS.

\bibliographystyle{plainnat}
\bibliography{references}
	
\end{document}